\documentclass[runningheads]{llncs}
\usepackage{graphicx,amsmath,amssymb,hyperref,subfig}
\usepackage[ruled, noline, algo2e, noend, linesnumbered]{algorithm2e}
\usepackage{cite}
\newcommand{\pathblast}{\textsc{PathBlast}\xspace}
\newcommand{\networkblast}{\textsc{NetworkBlast}\xspace}
\newcommand{\graal}{\textsc{Graal}\xspace }
\newcommand{\isorank}{\textsc{IsoRank}\xspace}
\newcommand{\natalie}{\textsc{natalie}\xspace}
\newcommand{\natnewversion}{\textsc{natalie 2.0}\xspace}
\newcommand{\mawish}{\textsc{MaWISh}\xspace}
\newcommand{\graemlin}{\textsc{Graemlin}\xspace }
\newcommand{\submap}{\textsc{SubMAP}\xspace}

\DeclareMathOperator*{\argmax}{arg\,max}
\newcommand{\ie}{i.e.\ }

\SetKw{KwAnd}{and}
\SetKw{KwOr}{or}
\newcommand{\obacht}[2]{}

\urldef{\mailmohammedgunnar}\path|{m.el-kebir,gunnar.klau}@cwi.nl|
\urldef{\mailjaap}\path|heringa@few.vu.nl|

\begin{document}

\title{Lagrangian Relaxation Applied to\\Sparse Global Network Alignment}
\author{
  Mohammed El-Kebir\inst{1,2,3} 
  \and Jaap Heringa\inst{2,3,4} 
  \and Gunnar W.\ Klau\inst{1,3}}

\titlerunning{Sparse Global Network Alignment}

\institute{
  Centrum Wiskunde \& Informatica, Life Sciences Group,
  Science Park 123, 1098~XG Amsterdam, the Netherlands, 
  \mailmohammedgunnar 
  \and Centre for Integrative Bioinformatics VU
  (IBIVU), VU University Amsterdam, De Boelelaan 1081A, 1081 HV
  Amsterdam, the Netherlands, \mailjaap 
  \and Netherlands Institute for
  Systems Biology, Amsterdam, the Netherlands
  \and Netherlands Bioinformatics Centre, Nijmegen, the Netherlands
}

\maketitle

\begin{abstract}
Data on molecular interactions is increasing at a tremendous
pace, while the development of solid methods for analyzing this network data is lagging
behind. This holds in particular for the field of comparative network
analysis, where one wants to identify commonalities between biological
networks.  Since biological functionality primarily operates at the
network level, there is a clear need for topology-aware comparison
methods.
In this paper we present a
method for global network alignment that is fast and robust, and can flexibly deal with
various scoring schemes taking both node-to-node correspondences as well as
network topologies into account. It is based on an integer linear
programming formulation, generalizing the well-studied quadratic
assignment problem. We obtain strong upper and
lower bounds for the problem by improving a Lagrangian relaxation
approach and introduce the software tool \natnewversion, a publicly available
implementation of our method. In an extensive
computational study on protein interaction networks for six different species,
we find that our new method outperforms alternative state-of-the-art methods with respect to
quality and running time.
  % TODO list:
  %\begin{itemize}
  %  \item Explain why sparse?
  %  \item Look at inapproximability result for MaxClique and whether we can say
  %  sth about it.
  %  \item Say sth about $Z_{LD} = Z_{LP}$
  %  \item Do we need to give an address for NISB affiliation
  %\end{itemize}
\end{abstract}

\section{Introduction}

% \begin{enumerate}
%   \item Exponential/tremendous increase in interaction data 
%   \item Need for comprehensive methods allowing for comparative analysis
%   \item We consider pairwise global network alignment
% \end{enumerate}

In the last decade, data on molecular interactions has increased at a tremendous
pace. For instance, the STRING database \cite{Szklarczyk2010}, which contains
protein protein interaction (PPI) data, grew from 261,033 proteins in 89
organisms in 2003 to 5,214,234 proteins in 1,133 organisms in May 2011, more
than doubling the number of proteins in the database every two years. The same
trends can be observed for other types of biological networks, including
metabolic, gene-regulatory, signal transduction and metagenomic networks, where
the latter can incorporate the excretion and uptake of organic compounds through,
for example, a microbial community \cite{Sharan2006,Kanehisa2006}.
In addition to the plethora of experimentally derived network data for many
species, also the structure and behavior of molecular networks have become
intensively studied over the last few years \cite{Alon2007}, leading to the
observation of many conserved features at the network level. However, the
development of solid methods for analyzing network data is lagging behind,
particularly in the field of comparative network analysis.
Here, one wants to identify commonalities between biological networks from
different strains or species, or derived form different conditions. Based
on the assumption that evolutionary conservation implies functional
significance, comparative approaches may help (i) improve the accuracy of data,
(ii) generate, investigate, and validate hypotheses, and (iii) transfer
functional annotations.
Until recently, the most common way of comparing two networks has been to
solely consider
node-to-node correspondences, for example by finding homologous
relationships between nodes (e.g. proteins in PPI networks) of either network,
while the topology of the two networks has not been taken into account. Since
biological functionality primarily operates at the network level, there is a
clear need for topology-aware comparison methods. In this paper we present a
network alignment method that is fast and robust, and can flexibly deal with
various scoring schemes taking both node-to-node correspondences as well as
network topologies into account.

\subsubsection*{Previous work.}
\label{sec:prev_work}
Network alignment establishes node correspondences based on both node-to-node
similarities and conserved topological information. Similar to sequence
alignment, \emph{local} network alignment aims at identifying one or more 
shared subnetworks, whereas \emph{global} network alignment addresses the
overall comparison of the complete input networks.

Over the last years a number of methods have been proposed for both global and local network
alignment, for example \pathblast \cite{Kelley2003}, \networkblast
\cite{Sharan2005}, %,Kalaev2008}, 
\mawish \cite{Koyuturk2006}, \graemlin \cite{Flannick2006}, 
\isorank \cite{Singh2008}, \graal \cite{Kuchaiev2010},  and \submap
\cite{Ay2011}. \pathblast heuristically computes
high-scoring similar paths in two PPI networks. Detecting protein
complexes has been addressed with \networkblast by Sharan et al.\ \cite{Sharan2005},
where the authors introduce a probabilistic model and propose a heuristic greedy approach to search for
shared complexes. Koyut\"urk et al.\ \cite{Koyuturk2006} use a more
elaborate scoring scheme based on an evolutionary model to compute
local pairwise alignments of PPI networks. % Narayanan and Karp [7] compare two PPI networks using a different model based on a graph-matching algorithm. They restrict the structural conservation to the environ- ment of a node and thus achieve a polynomial running time.
The \isorank algorithm by Singh et al.\ \cite{Singh2008} approaches
the global alignment problem by
preferably matching nodes which have a similar neighborhood, which
is elegantly solved as an eigenvalue problem. Kuchaiev et
al.\ \cite{Kuchaiev2010} take a similar approach. Their method \graal matches
nodes that share a similar distribution of so-called graphlets, which
are small connected non-isomorphic induced subgraphs. %\obacht[gunnar]{mohammed, i can't access the
%  graal paper, so this may be nonsense what I wrtie here.}

In this paper we focus on pairwise global network alignment, where an alignment
is scored by summing up individual scores of aligned node and interaction pairs.
Among the above mentioned methods, \isorank and \graal use a scoring model that
can be expressed in this manner.

\subsubsection*{Contribution.}
\label{sec:contribution}

We present an algorithm for global network alignment based on an integer linear
programming (ILP) formulation, generalizing the well-studied quadratic
assignment problem (QAP). We improve upon an existing Lagrangian relaxation
approach presented in previous work \cite{Klau2009} to obtain strong upper and
lower bounds for the problem. We exploit the closeness to QAP and generalize a
dual descent method for updating the Lagrangian multipliers to the generalized
problem. We have implemented the
revised algorithm from scratch as the software tool \natnewversion. In an extensive
computational study on protein interaction networks for six different species,
we compare \natnewversion to \graal and \isorank, evaluating the number of conserved
edges as well as functional coherence of the modules in terms of GO annotation.
We find that \natnewversion outperforms the alternative methods with respect to
quality and running time.
%%
%%sparse??
%%
Our software tool \natnewversion as well as all data sets used in this
study are publicly available at \url{http://planet-lisa.net}.

% \subsection{Previous work}
% \begin{enumerate}
%   \item \isorank
%   \item \GRAAL
%   \item QAP, more about that in Section~X
% \end{enumerate}

% \subsection{Our contribution}
% \begin{enumerate}
%   \item A new/improved method based on QAP, dual descent
%   \item Extensive experiments
% \end{enumerate}

\section{Preliminaries}

Given two simple graphs $G_1 = (V_1, E_1)$ and $G_2 = (V_2, E_2)$, an
\emph{alignment} ${a : V_1 \rightharpoondown V_2}$ is a \emph{partial injective
function} from $V_1$ to $V_2$. As such we have that an alignment relates every
node in $V_1$ to at most one node in $V_2$ and that conversely every node in
$V_2$ has at most one counterpart in $V_1$. An alignment is assigned a
real-valued \emph{score} using an additive scoring function $s$ defined as
follows:
\begin{equation}
\label{eq:alignment_score}
s(a) = \sum_{v \in V_1} c(v, a(v)) + \sum_{\substack{v,w \in V_1\\v<w}} w(v, a(v), w, a(w))
\end{equation}
where $c : V_1 \times V_2 \rightarrow \mathbb{R}$ is the score of aligning a
pair of nodes in $V_1$ and $V_2$ respectively. On the other hand, $w : V_1
\times V_2 \times V_1 \times V_2 \rightarrow \mathbb{R}$ allows for scoring
topological similarity. The problem of global pairwise network alignment (GNA)
is to find the highest scoring alignment $a^*$, i.e. $a^* = \argmax s(a)$. 
Figure~\ref{fig:example} shows an example.

\begin{figure}[btp]
  \center
  \includegraphics[scale=.8]{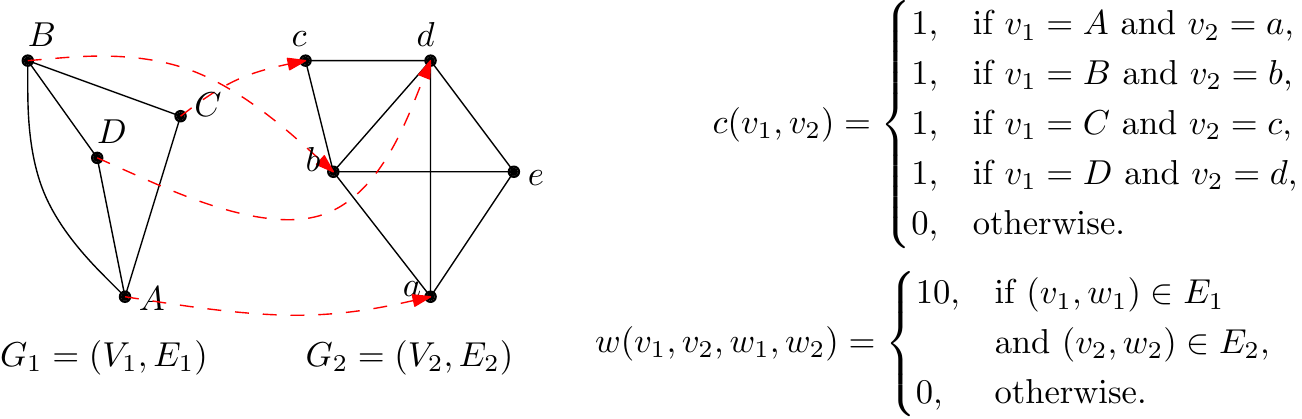}
  \caption{Example of a network alignment. With the given scoring function, the
  alignment has a score of $4+40=44$.}
  \label{fig:example}
\end{figure}

NP-hardness of GNA follows by a simple reduction from the decision problem
\textsc{Clique}, which asks whether there is a clique of cardinality at least
$k$ in a given simple graph $G = (V, E)$ \cite{Karp1972}. The corresponding GNA
instance concerns the alignment of the complete graph of $k$ vertices $K_k =
(V_k, E_k)$ with $G$ using the scoring function $s(a) = |\{ (v,w) \in E_k \mid
(a(v),a(w)) \in E\}|$. Since an alignment is injective, there is a clique of
cardinality at least $k$ if and only if the cost of the optimal alignment is
$\binom{k}{2}$.
%\obacht[mohammed]{I thought inapproximability result for
%MaxClique would carry over, but it's not so trivial, reduction does not preserve
%approximation factor, there is an amplification of $O(n^2)$. Don't know what to
%do here...}
The close relationship of GNA with the quadratic assignment problem is more easily
observed when formulating GNA as a mathematical program. Throughout the
remainder of the text we use dummy variables $i,j \in \{1, \ldots, |V_1|\}$ and
$k,l \in \{1, \ldots, |V_2|\}$ to denotes nodes in $V_1$ and $V_2$,
respectively. Let $C$ be a $|V_1| \times |V_2|$ matrix such that $c_{ik} =
c(i,k)$ and let $W$ be a $(|V_1| \times |V_2|) \times (|V_1| \times |V_2|)$
matrix whose entries $w_{ikjl}$ correspond to interaction scores $w(i,k,j,l)$.
Now we can formulate GNA as
\begin{alignat}{3}
\label{eq:obj_orig}
%Z_\mathrm{IP} \, = \, 
\tag{IQP}
\max_x      & \quad \sum_{i,k} c_{ik} x_{ik} + \sum_{\substack{i,j\\i < j}} 
                    \sum_{\substack{k,l\\k \neq l}} w_{ikjl} x_{ik} x_{jl}\\
\text{s.t.}
\label{eq:matching_j} & \quad \sum_{l} x_{jl} \leq 1 & \forall j\\
\label{eq:matching_l} & \quad \sum_{j} x_{jl} \leq 1 & \forall l\\
                      & \quad x_{ik}  \in \{0,1\} & \forall i,k
\end{alignat}
where the decision variable $x_{ik}$ indicates whether the $i$-th node in $V_1$
is aligned with the $k$-th node in $V_2$. The above formulation shares many
similarities with Lawler's formulation\cite{Lawler1963} of the QAP. However,
instead of finding an assignment we are interested in finding a matching, which
is reflected in constraints \eqref{eq:matching_j} and \eqref{eq:matching_l}
being inequalities rather than equalities. As can be seen in
\eqref{eq:alignment_score} we only consider the upper triangle of $W$
rather than the entire matrix. An analogous way of looking at this, is to
consider $W$ to be symmetric. This is usually not the case for QAP instances.
In addition, due to the fact that biological input graphs
are typically sparse, we have that $W$ is sparse as well. These differences allow us to
come up with an effective method of solving the problem as we will see in the
following.

\section{Method}
\label{sec:method}
The relaxation presented here follows the same lines as the one given by Adams
and Johnson for the QAP \cite{Adams1994}. We start by linearizing
\eqref{eq:obj_orig} by introducing binary variables $y_{ikjl}$
defined as $y_{ikjl} := x_{ik}x_{jl}$ and constraints $y_{ikjl} \leq x_{jl}$ and
$y_{ikjl} \leq x_{ik}$ for all $i \leq j$ and $k \neq l$. If we assume that all
entries in $W$ are positive, we do not need to enforce that $y_{ikjl} \geq
x_{ik} + x_{jl} - 1$. In Section~\ref{sec:Discussion} we will discuss this
assumption. Rather than using the aforementioned constraints, we make use of a
stronger set of constraints which we obtain by multiplying constraints
\eqref{eq:matching_j} and \eqref{eq:matching_l} by $x_{ik}$:

\begin{align}
\label{eq:lifted1}
\sum_{\substack{l\\l \neq k}} y_{ikjl} = \sum_{\substack{l\\l \neq k}} x_{ik}
x_{jl} \leq \sum_l x_{ik} x_{jl} \leq x_{ik}, & \quad \forall i,j,k, \; i < j\\
\label{eq:lifted2}
\sum_{\substack{j\\j > i}} y_{ikjl} = \sum_{\substack{j\\j > i}} x_{ik} x_{jl}
\leq \sum_j x_{ik} x_{jl} \leq x_{ik}, & \quad \forall i,k,l, \; k \neq l
\end{align}
We proceed by splitting the variable $y_{ikjl}$ (where $i < j$ and $k \neq l$).
In other words, we extend the objective function such that the counterpart of
$y_{ikjl}$ becomes $y_{jlik}$. This is accomplished by rewriting the dummy
constraint in \eqref{eq:lifted2} to $j \neq i$. In addition, we split the
weights: $w_{ikjl} = w_{jlik} = ({w^\prime_{ikjl}}/{2})$ where
$w^\prime_{ikjl}$ denotes the original weight. Furthermore, we require that the
counterparts of the split decision variables assume the same value, which
amounts to
\begin{alignat}{3}
%Z_\mathrm{LD} \, = \, 
\label{eq:ILP}
\tag{ILP} \max_{x,y}  & \quad \sum_{i,k} c_{ik} x_{ik} + \sum_{\substack{i,j\\i<j}}
\sum_{\substack{k,l\\k \neq l}} w_{ikjl} y_{ikjl} + \sum_{\substack{i,j\\i>j}}
\sum_{\substack{k,l\\k \neq l}} w_{ikjl} y_{ikjl}\hspace*{-6em}\\
\label{eq:ld_x1} \text{s.t.} & \quad \sum_{l} x_{jl} \leq 1 & \forall j\\
\label{eq:ld_x2}             & \quad \sum_{j} x_{jl} \leq 1 & \forall l\\
\label{eq:ld_y1}             & \quad \sum_{\substack{l\\l \neq k}} y_{ikjl} \leq x_{ik} & \forall i,j,k, \; i \neq j\\
\label{eq:ld_y2}             & \quad \sum_{\substack{j\\j \neq i}} y_{ikjl} \leq x_{ik} & \forall i,k,l, \; k \neq l\\
\label{eq:ld_eq} & \quad y_{ikjl} = y_{jlik}  & \forall i,j,k,l, \; i < j, k \neq l\\
\label{eq:ld_y}              & \quad y_{ikjl} \in \{0,1\} & \forall i,j,k,l, \;  i \neq j, k \neq l\\
\label{eq:ld_x}              & \quad x_{ik}  \in \{0,1\} & \forall i,k
\end{alignat}
We can solve the continuous relaxation of \eqref{eq:ILP} via its Lagrangian dual
by dualizing the linking constraints \eqref{eq:ld_eq} with multiplier $\lambda$:
\begin{alignat}{3}
\label{eq:LD}
\tag{LD} \min_\lambda & \quad Z_\mathrm{LD}(\lambda)\enspace,
\end{alignat}
where $Z_\mathrm{LD}(\lambda)$ equals
\begin{alignat}{3}
\nonumber
%Z_\mathrm{LD}(\lambda) = 
%& \, \max_{x,y} \quad \sum_{i,k} c_{ik} x_{ik} + \sum_{\substack{i,j\\i<j}} \sum_{\substack{k,l\\k \neq l}} w_{ikjl} y_{ikjl} + \sum_{\substack{i,j\\i>j}} \sum_{\substack{k,l\\k \neq l}} w_{ikjl} y_{ikjl}  + \sum_{\substack{i,j\\i<j}} \sum_{\substack{k,l\\k \neq l}} \lambda_{ikjl} (y_{ikjl} - y_{jlik})\\
                           & \, \max_{x,y} \quad \sum_{i,k} c_{ik} x_{ik} + \sum_{\substack{i,j\\i<j}} \sum_{\substack{k,l\\k \neq l}} (w_{ikjl} + \lambda_{ikjl}) y_{ikjl} + \sum_{\substack{i,j\\i>j}} \sum_{\substack{k,l\\k \neq l}} (w_{ikjl} - \lambda_{jlik}) y_{ikjl}\\
\nonumber & \,\,\,\, \text{s.t.} \quad  \mbox{\eqref{eq:ld_x1}, \eqref{eq:ld_x2}, \eqref{eq:ld_y1}, \eqref{eq:ld_y2}, \eqref{eq:ld_y} and \eqref{eq:ld_x}}
\end{alignat}
Now that the linking constraints have been dualized, one can observe that the
remaining constraints decompose the variables into $|V_1||V_2|$
disjoint groups, where variables across groups are not linked by any constraint, and where each
group contains a variable $x_{ik}$ and variables $y_{ikjl}$ for $j \neq i$ and
$l \neq k$. Hence, we have 
%Therefore we can solve $Z_\mathrm{LD}(\lambda)$ by solving a
%\emph{global} maximum weight bipartite matching problem where every edge weight
%is derived from a solution to a \emph{local} maximum weight bipartite matching
%problem. %The edge weights used in the global
\begin{alignat}{3}
\tag{LD$_\lambda$} Z_\mathrm{LD}(\lambda) \, = \, \max_x      & \quad \sum_{i,k} [c_{ik} + v_{ik}(\lambda)] x_{ik}\\
\text{s.t.} & \quad \sum_{l} x_{jl} \leq 1 & \quad \forall j\\
            & \quad \sum_{j} x_{jl} \leq 1 & \forall l\\
            & \quad x_{ik}  \in \{0,1\} & \quad\quad \forall i,k
\end{alignat}
which corresponds to a maximum weight bipartite matching problem on the so-called \emph{alignment graph} $G_m = (V_1 \cup V_2, E_m)$. In the general
case $G_m$ is a complete bipartite graph, \ie $E_m = \{(i, k) \mid i \in V_1, v_2
\in V_2\}$. However, by exploiting biological knowledge one can make $G_m$ more
sparse by excluding biologically-unlikely edges (see
Section~\ref{sec:Experiments}). For the global problem, the weight of a matching
edge $(i,k)$ is set to $c_{ik} + v_{ik}(\lambda)$, where the latter term is
computed as
%where% a global profit $v_{ik}$ is
\begin{alignat}{3}
\label{eq:localprob}
\tag{LD$^{ik}_\lambda$} v_{ik}(\lambda) \, = \, \max_{y} &
\sum_{\substack{j\\j>i}} \sum_{\substack{l\\l \neq k}} (w_{ikjl} +
\lambda_{ikjl}) y_{ikjl} + \sum_{\substack{j\\j<i}} \sum_{\substack{l\\l \neq
k}} (w_{ikjl} - \lambda_{jlik}) y_{ikjl}\hspace*{-3em}\\
\text{s.t.} & \quad \sum_{\substack{l\\l \neq k}} y_{ikjl} \leq 1 & \quad \forall j, \; j \neq i\\
            & \quad \sum_{\substack{j\\j \neq i}} y_{ikjl} \leq 1 & \quad \forall l, \; l \neq k\\
            & \quad y_{ikjl} \in \{0,1\} & \quad \forall j, l.
\end{alignat}
Again, this is a maximum weight bipartite matching problem on the same alignment
graph but excluding edges incident to either $i$ or $k$ and using different edge
weights: the weight of an edge $(j,l)$ is $w_{ikjl} + \lambda_{ikjl}$ if $j >
i$, or $w_{ikjl} - \lambda_{jlik}$ if $j < i$. 
So in order to compute $Z_\mathrm{LD}(\lambda)$, we need to solve a total number
of ${|V_1||V_2| + 1}$ maximum weight bipartite matching problems, which, using the
Hungarian algorithm \cite{Kuhn1953,Munkres1957} can be done in $O(n^5)$ time,
where $n = \max(|V_1|,|V_2|)$. In case the alignment graph is sparse, \ie
$O(|E_m|) = O(n)$, $Z_\mathrm{LD}(\lambda)$ can be computed in $O(n^4 \log n)$
time using the successive shortest path variant of the Hungarian algorithm
\cite{Edmonds1972}.
It is important to note that for any $\lambda$, $Z_\mathrm{LD}(\lambda)$ is an
upper bound on the score of an optimal alignment. This is because any alignment
$a$ is feasible to $Z_\mathrm{LD}(\lambda)$ and does not violate the original
linking constraints and therefore has an objective value equal to $s(a)$. In
particular, the optimal alignment $a^*$ is also feasible to
$Z_\mathrm{LD}(\lambda)$ and hence $a^* \leq Z_\mathrm{LD}(\lambda)$. Since the
two sets of problems resulting from the decomposition both have the integrality
property \cite{edmonds1965}, the smallest upper bound we can achieve equals the linear programming
(LP) bound of the continuous relaxation of \eqref{eq:ILP} \cite{Guignard2003}.
However, computing the smallest upper bound by finding suitable multipliers is
much faster than solving the corresponding LP.
Given solution $(x,y)$ to $Z_\mathrm{LD}(\lambda)$, we obtain a lower bound on
$s(a^*)$, denoted $Z_\mathrm{lb}(\lambda)$, by considering the score of the
alignment encoded in $x$.

\subsection{Solving Strategies}

In this section we will discuss strategies for identifying Lagrangian
multipliers $\lambda$ that yield an as small as possible gap between the upper
and lower bound resulting from the solution to $Z_\mathrm{LD}(\lambda)$.

\subsubsection{Subgradient optimization.} We start by discussing subgradient
optimization, which is originally due to Held and Karp \cite{Held1971}. The idea
is to generate a sequence  $\lambda^0, \lambda^1, \ldots$ of Lagrangian
multiplier vectors starting from $\lambda^0 = \mathbf{0}$ as follows:
\begin{align}
\lambda^{t+1}_{ikjl} = \lambda^t_{ikjl} - \frac{\alpha \cdot (Z_\mathrm{LD}(\lambda) -
Z_\mathrm{lb}(\lambda))}{\|g(\lambda^t)\|^2} g(\lambda^t_{ikjl}) & \quad\quad \forall
i,j,k,l, \; i < j, k \neq l
\end{align}
where $g(\lambda^t_{ikjl})$ corresponds to the subgradient of multiplier
$\lambda^t_{ikjl}$, \ie $g(\lambda^t_{ikjl}) = y_{ikjl} - y_{jlik}$, and $\alpha$
is the step size parameter. Initially $\alpha$ is set to $1$ and it is halved if
neither  $Z_\mathrm{LD}(\lambda)$ nor $Z_\mathrm{lb}(\lambda)$ have improved for
over $N$ consecutive iterations. Conversely, $\alpha$ is doubled if $M$ times in
a row there was an improvement in either $Z_\mathrm{LD}(\lambda)$ or
$Z_\mathrm{lb}(\lambda)$ \cite{Caprara1999}. In case all subgradients are zero, the optimal
solution has been found and the scheme terminates. Note that this is not
guaranteed to happen. Therefore we abort the scheme after exceeding a time limit or a
pre-specified number of iterations. In addition, we terminate if $\alpha$ has
dropped below machine precision. Algorithm~\ref{alg:subgradient} gives
the pseudo code of this procedure.

\begin{algorithm2e}
\caption{\textsc{SubgradientOpt}$(\lambda, M, N)$}
\label{alg:subgradient}
$\alpha \leftarrow 1$; $n \leftarrow N$; $m \leftarrow M$\\% $\lambda \leftarrow 0$\\
$[\mathrm{LB}^*, \mathrm{UB}^*] \leftarrow [Z_\mathrm{lb}(\lambda),
  Z_\mathrm{LD}(\lambda)]$\\
\While{$g(\lambda) \neq 0$}
{
  $\lambda \leftarrow \lambda - \frac{\alpha(Z_\mathrm{LD}(\lambda) -
Z_\mathrm{lb}(\lambda))}{\|g(\lambda^t)\|^2} g(\lambda^t)$\\
  \lIf {$[\mathrm{LB}^*, \mathrm{UB}^*] \setminus [Z_\mathrm{lb}(\lambda),
  Z_\mathrm{LD}(\lambda)] = \emptyset$}
  {
    $n \leftarrow n - 1$
  }\\
  \Else
  {
    $\mathrm{LB}^* \leftarrow \max[\mathrm{LB}^*, Z_\mathrm{lb}(\lambda)]$\\
    $\mathrm{UB}^* \leftarrow \min[\mathrm{UB}^*, Z_\mathrm{LD}(\lambda)]$\\
    $m \leftarrow m - 1$
  }
  \lIf {$n = 0$} {$\alpha \leftarrow \alpha / 2$; $n \leftarrow N$}\\
  \lIf {$m = 0$} {$\alpha \leftarrow 2\alpha$; $m \leftarrow M$}
}
\Return $[\mathrm{LB}^*, \mathrm{UB}^*]$
\end{algorithm2e}

\subsubsection{Dual descent.}%%
In this section we derive a dual descent method which is an extension of the one
presented in \cite{Adams1994}. The dual descent method takes as a starting point
the dual of $Z_\mathrm{LD}(\lambda)$:%%
\begin{alignat}{3}
\label{eq:dual_ld_global}
Z_\mathrm{LD}(\lambda) \, = \, \min_{\alpha,\beta} & \quad \sum_{i} \alpha_i + \sum_k \beta_k\\
\text{s.t.} & \quad \alpha_i + \beta_k \geq c_{ik} + v_{ik}(\lambda) & \quad\quad\quad \forall i, k\\
            & \quad \alpha_i \geq 0 & \forall i\\
            & \quad \beta_k \geq 0 & \forall k
\end{alignat}
where the dual of $v_{ik}(\lambda)$ is
\begin{alignat}{3}
\label{eq:dual_ld_local}
v_{ik}(\lambda) \, = \, \min_{\mu,\nu} & \quad \sum_{\substack{j\\j \neq i}} \mu^{ik}_j + \sum_{\substack{l\\l \neq k}} \nu^{ik}_l\\
\label{eq:dual_ld_local_1}
\text{s.t.} & \quad \mu^{ik}_j + \nu^{ik}_l \geq w_{ikjl} + \lambda_{ikjl} & \quad\quad\quad \forall j,l, \; j > i, l \neq k\\
\label{eq:dual_ld_local_2}
            & \quad \mu^{ik}_j + \nu^{ik}_l \geq w_{ikjl} - \lambda_{jlik} & \quad \forall j,l, \; j < i, l \neq k\\
\label{eq:dual_ld_local_3}
            & \quad \mu^{ik}_j \geq 0 & \forall j\\
\label{eq:dual_ld_local_4}
            & \quad \nu^{ik}_l \geq 0 & \forall l.
\end{alignat}
Suppose that for a given $\lambda^t$ we have computed dual variables $(\alpha,
\beta)$ solving \eqref{eq:dual_ld_global} with objective value
$Z_\mathrm{LD}(\lambda^t)$, as well as dual variables $(\mu^{ik}, \nu^{ik})$
yielding values $v_{ik}(\lambda)$ to linear programs \eqref{eq:dual_ld_local}.
The goal now is to find $\lambda^{t+1}$ such that the resulting bound is better
or just as good, i.e.\ ${Z_\mathrm{LD}(\lambda^{t+1}) \leq
Z_\mathrm{LD}(\lambda^t)}$. We prevent the bound from increasing, by ensuring
that the dual variables $(\alpha, \beta)$ remain feasible to
\eqref{eq:dual_ld_global}. This we can achieve by considering the slacks:
$\pi_{ik}(\lambda) = \alpha_i + \beta_k - c_{ik} - v_{ik}(\lambda).$
So for $(\alpha,\beta)$ to remain feasible, we can only allow every
$v_{ik}(\lambda^t)$ to increase by as much as
$\pi_{ik}(\lambda^t)$. We can achieve such an increase by considering
linear programs \eqref{eq:dual_ld_local} and their slacks defined as 
\begin{align}
  \gamma_{ikjl}(\lambda) & = \begin{cases} \mu^{ik}_j + \nu^{ik}_l - w_{ikjl} +
  \lambda_{ikjl}, & \mbox{if $j > i$,}\\
                                           \mu^{ik}_j + \nu^{ik}_l - w_{ikjl} -
                                           \lambda_{jlik}, & \mbox{if $j <
                                           i$,}\end{cases} & \forall j,l, \; j
                                           \neq i, l \neq k,
\end{align}
and update the multipliers in the following way.
\begin{lemma}
\label{lem:dualdescent}
The adjustment scheme below yields solutions to linear programs \eqref{eq:dual_ld_local} with objective values $v_{ik}(\lambda^{t+1})$ at most $\pi_{ik}(\lambda^t) + v_{ik}(\lambda^t)$ for all $i,k$.
\begin{equation}
  \label{eq:mult_update}
  \begin{split}
    \lambda^{t+1}_{ikjl} = \lambda^t_{ikjl} & + \varphi_{ikjl}\left[ \gamma_{ikjl}(\lambda^t) + \tau_{ik} \left( \frac{1}{2(n_1 - 1)} + \frac{1}{2(n_2 - 1)}\right) \pi_{ik}(\lambda^t)\right]\\
                                            & - \varphi_{jlik}\left[ \gamma_{jlik}(\lambda^t) + \tau_{jl} \left( \frac{1}{2(n_1 - 1)} + \frac{1}{2(n_2 - 1)}\right) \pi_{jl}(\lambda^t)\right]
  \end{split}
\end{equation}
for all $j,l, \, i<j, k \neq l$, where $n_1 = |V_1|$, $n_2 = |V_2|$, and $0 \leq \varphi_{ikjl}, \tau_{jl} \leq 1$ are parameters. 
\end{lemma}

\begin{proof}
We prove the lemma by showing that for any $i,k$ there exists a feasible solution $(\mu^{\prime ik},\nu^{\prime ik})$ to \eqref{eq:dual_ld_local} whose objective value $v_{ik}(\lambda^{t+1})$ is at most $\pi_{ik}(\lambda^t) + v_{ik}(\lambda^t)$. Let $(\mu^{ik}, \nu^{ik})$ be the solution to \eqref{eq:dual_ld_local} given multipliers $\lambda^t$. We claim that setting
\begin{alignat*}{3}
\mu^{\prime ik}_j = \mu^{ik}_j + \frac{\pi_{ik}(\lambda^t)}{2(n_1 - 1)} & \quad\quad \forall j,\;  j \neq i\\
\nu^{\prime ik}_l = \nu^{ik}_j + \frac{\pi_{ik}(\lambda^t)}{2(n_2 - 1)} & \quad\quad \forall l,\;  l \neq k,
\end{alignat*}
results in a feasible solution to \eqref{eq:dual_ld_local} given multipliers $\lambda^{t+1}$. We start by showing that constraints \eqref{eq:dual_ld_local_1} and \eqref{eq:dual_ld_local_2} are satisfied. From \eqref{eq:mult_update} the following bounds on $\lambda^{t+1}$ follow.
\begin{align*}
\lambda^t_{ikjl} - \gamma_{jlik}(\lambda^t) - \left( \frac{1}{2(n_1 - 1)} +
\frac{1}{2(n_2 - 1)}\right) \pi_{jl}(\lambda^t) \leq \lambda^{t+1}_{ikjl} & \;\quad \forall j,l, \; j < i, l \neq k\\
\lambda^{t+1}_{ikjl} \leq \lambda^t_{ikjl} + \gamma_{ikjl}(\lambda^t) + \left(
\frac{1}{2(n_1 - 1)} + \frac{1}{2(n_2 - 1)}\right) \pi_{ik}(\lambda^t) & \;\quad \forall j,l, \; j < i, l \neq k.
\end{align*}
Therefore we have that the following inequalities imply constraints
\eqref{eq:dual_ld_local_1} and \eqref{eq:dual_ld_local_2} for all
$j,l$, $j > i$, $l \neq k$:
\begin{alignat*}{3}
\mu^{\prime ik}_j + \nu^{\prime ik}_l \geq w_{ikjl} + \lambda^t_{ikjl} + \gamma_{ikjl}(\lambda^t) + \left( \frac{1}{2(n_1 - 1)} + \frac{1}{2(n_2 - 1)}\right) \pi_{ik}(\lambda^t)
\end{alignat*}
and for all $j,l$, $j < i$, $l \neq k$
\begin{alignat*}{3}
\mu^{\prime ik}_j + \nu^{\prime ik}_l \geq w_{ikjl} - \lambda^t_{jlik} + \gamma_{ikjl}(\lambda^t) + \left( \frac{1}{2(n_1 - 1)} + \frac{1}{2(n_2 - 1)}\right) \pi_{ik}(\lambda^t).
\end{alignat*}
Constraints \eqref{eq:dual_ld_local_1} and \eqref{eq:dual_ld_local_2}
are indeed implied, as,  for all $j,l$, $j > i$, $l \neq k$,
\begin{alignat*}{3}
\mu^{\prime ik}_j + \nu^{\prime ik}_l & = \mu^{ik}_j + \nu^{ik}_l + \left( \frac{1}{2(n_1 - 1)} + \frac{1}{2(n_2 - 1)}\right) \pi_{ik}(\lambda^t)\\
                                      & \geq w_{ikjl} + \lambda^t_{ikjl} + \gamma_{ikjl}(\lambda^t) + \left( \frac{1}{2(n_1 - 1)} + \frac{1}{2(n_2 - 1)}\right) \pi_{ik}(\lambda^t)
\end{alignat*}
and for all $j,l$, $j < i$, $l \neq k$
\begin{alignat*}{3}
\mu^{\prime ik}_j + \nu^{\prime ik}_l & = \mu^{ik}_j + \nu^{ik}_l + \left( \frac{1}{2(n_1 - 1)} + \frac{1}{2(n_2 - 1)}\right) \pi_{ik}(\lambda^t) \\
                                      & \geq w_{ikjl} - \lambda^t_{jlik} + \gamma_{ikjl}(\lambda^t) + \left( \frac{1}{2(n_1 - 1)} + \frac{1}{2(n_2 - 1)}\right) \pi_{ik}(\lambda^t).
\end{alignat*}
Since $\mu^{ik}_j, \nu^{ik}_l \geq 0$ ($\forall j,l, \; j \neq i, l \neq k$) and by definition $\pi_{ik}(\lambda^t) \geq 0$, constraints \eqref{eq:dual_ld_local_3} and \eqref{eq:dual_ld_local_4} are satisfied as well. 
The objective value of $(\mu^{\prime ik},\nu^{\prime ik})$ is given by
\begin{equation*}
\sum_{\substack{j\\j \neq i}} \mu^{\prime ik}_j + \sum_{\substack{l\\l \neq k}} \nu^{\prime ik}_l = \sum_{\substack{j\\j \neq i}} \mu^{ik}_j + \sum_{\substack{l\\l \neq k}} \nu^{ik}_l + \pi_{ik}(\lambda^t) = v_{ik}(\lambda^t) + \pi_{ik}(\lambda^t).
\end{equation*}
Since \eqref{eq:dual_ld_local} are minimization problems and there exist, for all $i,k$, feasible solutions with objective values $v_{ik}(\lambda^t) + \pi_{ik}(\lambda^t)$, we can conclude that the objective values of the solutions are bounded by this quantity. Hence, the lemma follows.
\end{proof}

We use $\varphi = 0.5$, $\tau =
1$, and perform the dual descent method $L$ successive times (see Algorithm~\ref{alg:dualdescent}).

\begin{algorithm2e}
\caption{\textsc{DualDescent}$(\lambda, L)$}
\label{alg:dualdescent}
$\varphi \leftarrow 0.5$; 
$[\mathrm{LB}^*, \mathrm{UB}^*] \leftarrow [Z_\mathrm{lb}(\lambda),
  Z_\mathrm{LD}(\lambda)]$\\
\For {$n \leftarrow 1$ \KwTo $L$}
{
  \ForEach{$i,k,j,l, \, i < j, k \neq l$}
  {
    $\lambda_{ikjl} \leftarrow \lambda_{ikjl} + \varphi(\gamma_{ikjl} +
    \frac{\pi_{ik}(\lambda)}{2(n_1 - 1)} + \frac{\pi_{ik}(\lambda)}{2(n_2 -
    1)})) - \varphi(\gamma_{jlik} +
    \frac{\pi_{jl}(\lambda)}{2(n_1 - 1)} + \frac{\pi_{jl}(\lambda)}{2(n_2 -
    1)}))$
  }
  $\mathrm{LB}^* \leftarrow \max[\mathrm{LB}^*, Z_\mathrm{lb}(\lambda)]$\\
  $\mathrm{UB}^* \leftarrow Z_\mathrm{LD}(\lambda)$\\
}
\Return $[\mathrm{LB}^*, \mathrm{UB}^*]$
\end{algorithm2e}

%The following corollary follows immediately from the previous lemma.
%\begin{corollary}
%Updating the multipliers according to \eqref{eq:mult_update} results in a sequence of non-increasing upper bounds, i.e.\ $Z_\mathrm{LD}(\lambda^{t+1}) \leq Z_\mathrm{LD}(\lambda^t)$ for all $t \in \mathbb{N}$.
%\end{corollary}

\subsubsection{Overall method.}
Our overall method combines both the subgradient optimization and dual descent
method. We do this performing the subgradient method until termination and then
switching over to the dual descent method. This procedure is repeated
$K$ times (see Algorithm~\ref{alg:gnaopt}).

\begin{algorithm2e}
\caption{\textsc{Natalie}$(K, L, M, N)$}
\label{alg:gnaopt}
$\lambda \leftarrow \mathbf{0}$; 
$[\mathrm{LB}^*, \mathrm{UB}^*] \leftarrow [0, \infty]$\\
\For {$k \leftarrow 1$ \KwTo $K$}
{
  $[\mathrm{LB}^*, \mathrm{UB}^*] \leftarrow 
    \textsc{SubgradientOpt}(\lambda, M,  N) \cap [\mathrm{LB}^*, \mathrm{UB}^*]$\\
  $[\mathrm{LB}^*, \mathrm{UB}^*] \leftarrow 
    \textsc{DualDescent}(\lambda, L) \cap [\mathrm{LB}^*, \mathrm{UB}^*]$\\
}
\Return $[\mathrm{LB}^*, \mathrm{UB}^*]$
\end{algorithm2e}
%%\subsection{Implementation}

We implemented \natalie in C\raisebox{0.5ex}{\small ++} using the LEMON graph
library (\url{http://lemon.cs.elte.hu/}).
%%\footnote[1]{\url{http://lemon.cs.elte.hu/}}. 
The successive shortest path
algorithm for maximum weight bipartite matching was implemented and contributed
to LEMON. Special care was taken to deal with the inherent numerical instability
of floating point numbers. Our implementation supports both the GraphML and GML
graph formats. Rather than using one big alignment graph, we store and use a
different alignment graph for every local problem \eqref{eq:localprob}.
This proved to be a huge improvement in running times, especially when the global alignment graph
is sparse. \natalie is publicly available at \url{http://planet-lisa.net}.

\section{Experimental Evaluation}
\label{sec:Experiments}

From the STRING database v8.3 \cite{Szklarczyk2010}, we obtained PPI networks
for the following six species:
\emph{C.~elegans} (cel), \emph{S.~cerevisiae} (sce), \emph{D.~melanogaster}
(dme), \emph{R.~norvegicus} (rno), \emph{M.~musculus} (mmu) and
\emph{H.~sapiens} (hsa). We only considered interactions that were
experimentally verified. Table~\ref{tbl:input} shows the sizes of the networks. 
%With \texttt{ggsearch36} provided by the FASTA v36
%package\footnote[2]{\url{http://faculty.virginia.edu/wrpearson/fasta/fasta36/}},
We performed, using the BLOSUM62 matrix, an all-against-all global sequence
alignment on the protein sequences of all $\binom{6}{2} = 15$ pairs of networks. We used affine gap
penalties with a gap-open penalty of 2 and a gap-extension penalty of 10. The
first experiment in Section~\ref{sec:edgecorrectness} compares the raw
performance of \isorank, \graal and \natalie in terms of objective value. In
Section~\ref{sec:go} we evaluate the biological relevance of the alignments produced by
the three methods. 
%We conclude by showing how clustering based on scores
%including both sequence and topological information yields the correct
%phylogeny. 
All experiments were conducted on a compute cluster with 2.26 GHz
processors with 24 GB of RAM.

\begin{table}
  \center
  \begin{tabular}{lrrr}
    \hline
    \textbf{species} & ~\textbf{nodes}~ & ~\textbf{annotated}~ & ~\textbf{interactions}\\
    \hline
    cel (c) & 5,948   & 4,694 & 23,496\\
    sce (s) & 6,018  & 5,703& 131,701 \\
    dme (d) & 7,433 & 6,006 & 26,829  \\
    rno (r) & 8,002  & 6,786& 32,527  \\
    mmu (m) & 9,109  & 8,060& 38,414  \\
    hsa (h) & 11,512 & 9,328 & 67,858  \\
   %% \hline
  \end{tabular}
  \caption{Characteristics of input networks considered in this
    study. The columns contain species identifier, number of nodes in
    the network, number of annotated nodes thereof, and number of interactions }
  \label{tbl:input}
\end{table}

\subsection{Edge-Correctness}
\label{sec:edgecorrectness}

The objective function used for scoring alignments in \graal counts the number
of mapped edges. Such an objective function is easily expressible in our
framework using $s(a) = |\{ (v,w) \in E_1 \mid (a(v),a(w)) \in E_2\}|$
and can also be modeled using the \isorank scoring function. In order to compare performance of the methods
across instances, we normalize the scores by dividing by $\min(|E_1|, |E_2|)$.
This measure is called the edge-correctness by Kuchaiev et al.\ \cite{Kuchaiev2010}.

As mentioned in Section~\ref{sec:method}, our method
benefits greatly from using a sparse alignment
graph.
%\obacht[gunnar]{sparsity story ois not so clear here...}. 
To that end, we use the e-values
obtained from the all-against-all sequence alignment to prohibit biologically
unlikely matchings by only considering protein-pairs
whose e-value is at most 100. Note that this only applies to \natalie as both
\graal and \isorank consider the complete alignment graph. On each of the 15
instances, we ran \graal with 3 different random seeds and sampled the input
parameter which balances the contribution of the graphlets with the node
degrees uniformly within the allowed range of $[0,1]$. As for \isorank, when
setting the parameter $\alpha$|which controls to what extent topological similarity plays
a role|to the desired value of 1, very poor
results were obtained. Therefore we also sampled this parameter within its
allowed range and re-evaluated the resulting alignments in terms of
edge-correctness. \natalie was run with a time limit of 10 minutes and $K=3$, 
$L=100$, $M=10$, $N=20$. For both \graal and \isorank only the highest-scoring results were
considered.
\begin{figure}[tbp]
  \center
  \subfloat[Edge correctness]{\includegraphics[scale=0.8]{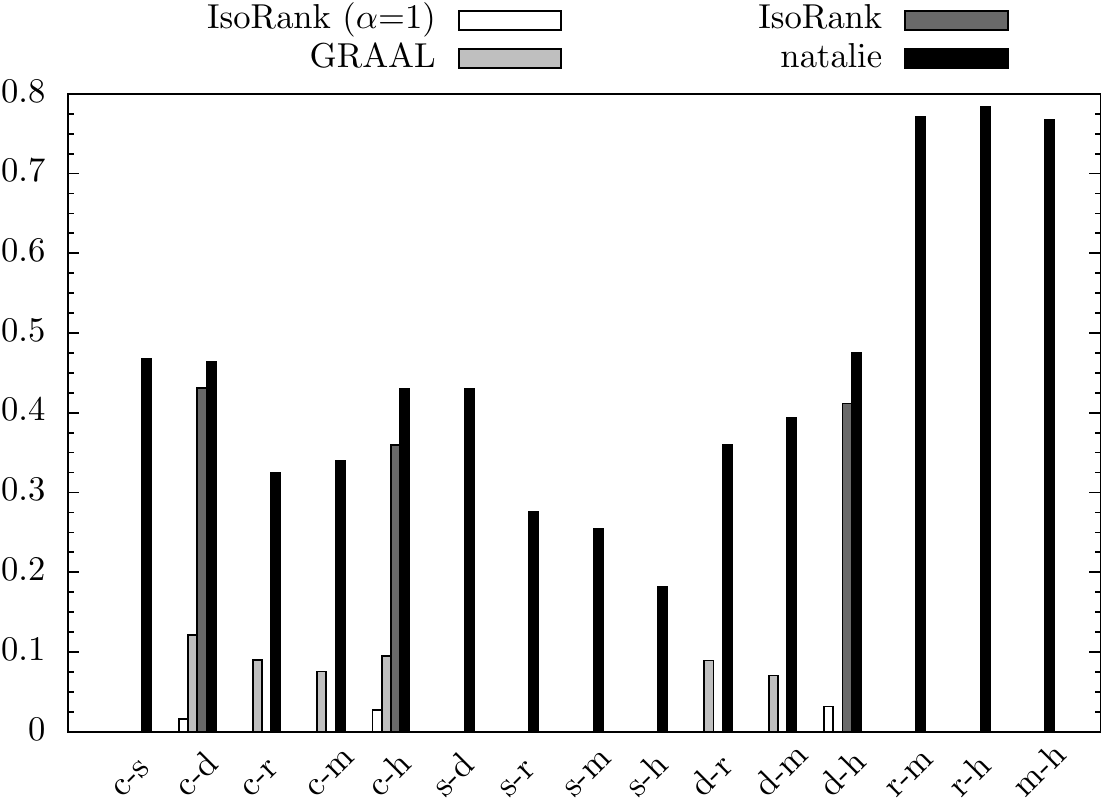}}\\
  %\hspace{.1cm} 
  \subfloat[Running times in seconds]{\includegraphics[scale=.8]{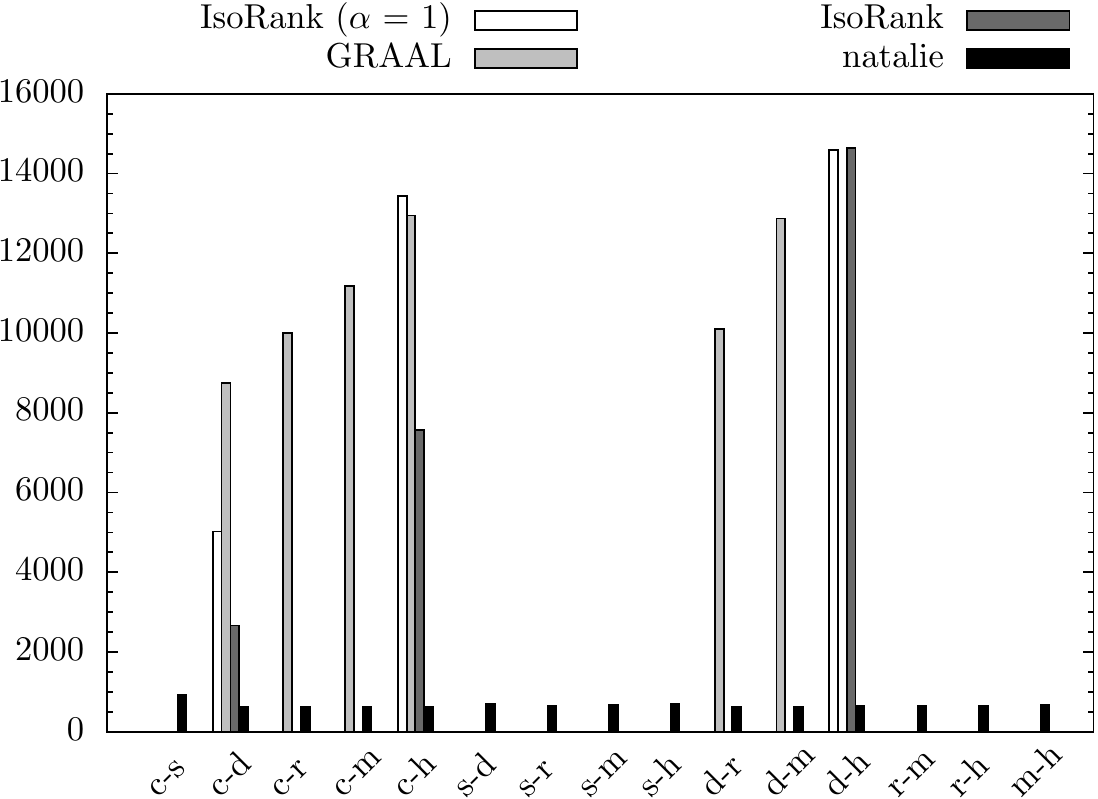}}
  \caption{Performance of the three different methods for the
    all-against-all species comparisons (15 alignment
    instances). Missing bars correspond to exceeded time/memory limits or software crashes.}
  \label{fig:ec}
\end{figure}

Figure~\ref{fig:ec} shows the results. \isorank was only able to compute alignments for three out of the 15 instances.
On the other instances \isorank crashed, which may be due to the large size of
the input networks. For \graal no alignments concerning \emph{sce} could be
computed, which is due to the large number of edges in the network
on which the graphlet enumeration procedure choked: in 12 hours only for
3\% of the nodes the graphlet degree vector was computed. As for the last three
instances, \graal crashed due to exceeding the memory limit inherent to 32-bit processes.
Unfortunately no 64-bit executable was available. On
the instances for which \graal could compute alignments, the performance|both in
solution quality and running time|is
very poor when compared to \isorank and \natalie. \natalie outperforms \isorank
in both running time and solution quality.
%In order to investigate
%this further, we aligned the same yeast2-human1 instance of \cite{Kuchaiev2010}
%using \natalie. By simply randomly generating a sparse
%alignment graph we were already able to outperform \graal. We believe that the
%poor performance of \isorank on yeast2-human1 is due to the balance parameter
%being set to the expected value of 1.

\subsection{GO Similarity}
\label{sec:go}
In order to measure the biological relevance of the obtained network alignments,
we make use of the Gene Ontology (GO) \cite{Ashburner2000}. For every node in
each of the six networks we obtained a set of GO annotations (see
Table~\ref{tbl:input} for the exact numbers). Each annotation set was extended
to a multiset by including all ancestral GO terms for every annotation in the
original set. Subsequently we employed a similarity measure that
compares a pair of aligned nodes based on their GO annotations and also takes
into account the relative frequency of each annotation \cite{Jaeger2010}.
%(see Appendix~\ref{sec:gosim} for details)\obacht[mohammed]{Don't know if I have time for this}
Since the similarity measure assigns a score between 0 and 1 to every aligned
node pair, the highest similarity score one can get for any alignment is the
minimum number of annotated nodes in either of the networks.
Therefore we can normalize the similarity scores by this quantity.   
Unlike the previous experiment, this time we considered the bitscores of
the pairwise global sequence alignments. Similarly to \isorank parameter
$\alpha$, we introduced a parameter $\beta \in [0,1]$ such that the sequence part of the
score has weight $(1 - \beta)$ and the topology part has weight $\beta$.
For both \isorank and \natalie we sampled the weight parameters uniformly in the
range $[0,1]$ and showed the best result in Figure~\ref{fig:go}. There we can
see that both \natalie and \isorank identify functionally coherent alignments.

\begin{figure}[btp]
  \center
  \includegraphics[scale=.8]{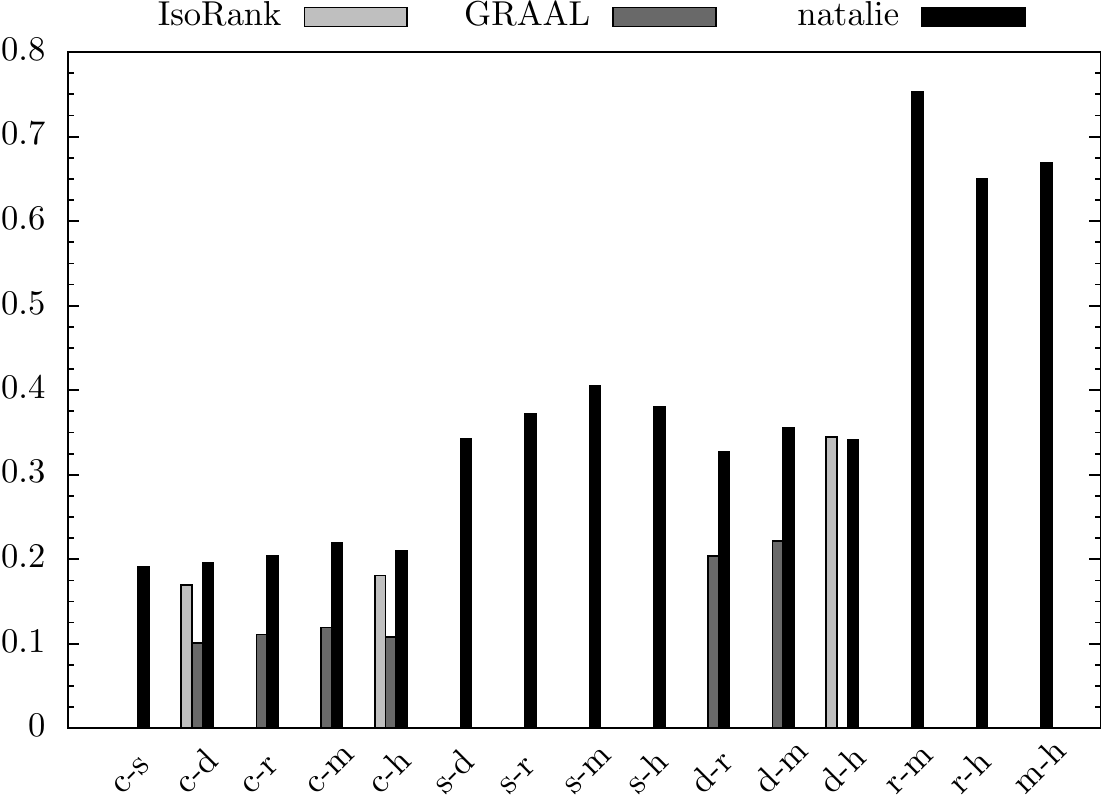}
  \caption{Biological relevance of the alignments measured via GO similarity}
  \label{fig:go}
\end{figure}

\section{Conclusion}
\label{sec:Conclusion}
\label{sec:Discussion}

Inspired by results for the closely related quadratic assignment
problem (QAP), we have presented new algorithmic ideas in order to
make a Lagrangian relaxation approach for global network alignment
practically useful and competitive. In particular, we have
generalized a dual descent method for the QAP. We have
found that combining this scheme with the traditional subgradient
optimization method leads to fastest progress of upper and lower
bounds.

Our implementation of the new method, \natnewversion, works very well
and fast when aligning biological networks, which we have shown in an
extensive study on the alignment of cross-species PPI networks. We have compared
\natnewversion to those state-of-the-art methods whose scoring
schemes can be expressed as special cases of the scoring scheme
we propose. Currently, these methods are \isorank and \graal. Our
experiments show that the Lagrangian relaxation approach is a very
powerful method and that it currently outperforms the competitors in
terms of quality of the results and running time. 

Currently, all methods, including ours, approach the global network
alignment problem heuristically, that is, the computed alignments are
not guaranteed to be optimal solutions of the problem. While the
other approaches are intrinsically heuristic---both \isorank and
\graal, for instance, approximate the neighborhood of a node and then
match it with a similar node---the inexactness in our methods has
two causes that we plan to address in future work: On the one hand,
there may still be a gap between upper and lower bound of the
Lagrangian relaxation approach after the last iteration. We can use
these bounds, however, in a branch-and-bound approach that will
compute provably optimal solutions.  On the other hand, we currently
do not consider the complete bipartite alignment graph and may
therefore miss the optimal alignment. Here, we will investigate
preprocessing strategies, in the spirit of
\cite{wohlers_proteins-2011}, to safely sparsify the input bipartite graph
without violating optimality conditions.

The independence of the local problems \eqref{eq:localprob} allows for easy
parallelization, which, when exploited would lead to an even faster method.
Another improvement in running times might be achieved when considering more
involved heuristics for computing the lower bound, such as local
search. 
%the local search procedure proposed in \cite{Chindelevitch2010}. 
More functionally-coherent
alignments can be obtained when considering a scoring function where
node-to-node correspondences are not only scored via sequence similarity but
also for instance via GO similarity. In certain cases, even negative weights for
topological interactions might be desired in which case one needs to reconsider
the assumption of entries of matrix $W$ being positive.

\begin{small}
\paragraph{Acknowledgments.} We thank SARA Computing and Networking Services
(\url{www.sara.nl}) for their support in using the Lisa Compute Cluster. In
addition, we are very grateful to Bernd Brandt for helping out with various
bioinformatics issues and also to Samira Jaeger for providing code and advice on
the GO similarity experiments.
\end{small}

\bibliographystyle{abbrv}
\bibliography{gna}

%%\appendix
%%\clearpage

%%\section{Proof of Lemma~\ref{lem:dualdescent}}
%%\label{sec:proof}

\end{document}